\PassOptionsToPackage{table}{xcolor}
\documentclass[a4paper, 11pt]{article}

\title{Graph neural networks and MSO} 

\date{} 					

\usepackage{authblk}

\author[ \hspace{-1ex}]{Veeti Ahvonen}
\author[ \hspace{-1ex}]{Damian Heiman}
\author[ \hspace{-1ex}]{Antti Kuusisto}

\affil[ \hspace{-1ex}]{Mathematics Research Centre, Tampere University, Finland}

\usepackage{graphicx} 

\usepackage[textwidth=151mm, hcentering, vmargin={24mm,32.5mm}, foot=20mm]{geometry}

\usepackage{hyperref}
\usepackage{verbatim}

\usepackage[utf8]{inputenc}
\usepackage[T1]{fontenc}
\usepackage{amsmath, amssymb, amsthm}
\usepackage{enumitem}
\setitemize{noitemsep,topsep=0pt,parsep=3pt,partopsep=0pt}
\setenumerate{noitemsep,topsep=0pt,parsep=3pt,partopsep=0pt}
\usepackage{tikz}
\usepackage{lipsum}
\usepackage{pdfpages}
\usepackage{multicol} 
\usepackage{stmaryrd}
\usepackage{caption}
\usepackage{soul}
\usepackage{makecell}
\usepackage{booktabs}

\usepackage{arydshln}
\usepackage{makecell}
\usepackage{mathtools}

\usepackage{marginnote}

\usepackage{colonequals}

\usepackage{multirow}

\theoremstyle{plain}
\newtheorem{theorem}{Theorem}[section]
\newtheorem{lemma}[theorem]{Lemma}

\newtheorem{proposition}[theorem]{Proposition}

\theoremstyle{definition}

\newcommand{\N}{\mathbb N}
\newcommand{\Z}{\mathbb Z}

\newcommand{\R}{\mathbb R}

\newcommand{\be}{\mathbf{e}}

\newcommand{\bs}{\mathbf{s}}

\newcommand{\cA}{\mathcal{A}}
\newcommand{\cB}{\mathcal{B}}
\newcommand{\cC}{\mathcal{C}}

\newcommand{\cF}{\mathcal{F}}
\newcommand{\cG}{\mathcal{G}}

\newcommand{\cL}{\mathcal{L}}
\newcommand{\cM}{\mathcal{M}}

\newcommand{\cP}{\mathcal{P}}

\newcommand{\cR}{\mathcal{R}}

\newcommand{\cT}{\mathcal{T}}

\newcommand{\cX}{\mathcal{X}}
\newcommand{\cY}{\mathcal{Y}}

\newcommand*{\abs}[1]{\lvert#1\rvert}   

\newcommand{\msc}{\mathrm{MSC}}

\newcommand{\GMSC}{\mathrm{GMSC}}

\newcommand{\prop}{\mathrm{PROP}}
\newcommand{\var}{\mathrm{VAR}}

\newcommand{\GNN}{\mathrm{GNN}}

\newcommand{\MSO}{\mathrm{MSO}}
\newcommand{\GML}{\mathrm{GML}}
\newcommand{\VGML}{\omega\text{-}\mathrm{GML}}
\newcommand{\FCMPA}{\mathrm{FCMPA}}
\newcommand{\CMPA}{\mathrm{CMPA}}
\newcommand{\AGG}{\mathrm{AGG}}
\newcommand{\COM}{\mathrm{COM}}
\newcommand{\GNNF}{\mathrm{GNN[F]}}

\date{}

\begin{document}

\maketitle

\begin{abstract}
    \noindent
    We give an alternative proof for the existing result that recurrent graph neural networks working with reals have the same expressive power in restriction to monadic second-order logic MSO as the graded modal substitution calculus. The proof is based on constructing distributed automata that capture all MSO-definable node properties over trees.
    We also consider some variants of the acceptance conditions.
\end{abstract}

\section{Introduction}

Graph neural networks (GNNs) are a deep learning architecture for processing graph-structured data, and they have proven useful in various applications. They operate in labeled graphs and assign feature vectors (a vector of reals) to nodes in discrete rounds. Feature updates are carried out in each node by aggregating the feature vectors of the node's neighbours into a single vector and then combining it with the node's own feature vector. The result we prove in this article concerns recurrent graph neural networks; while a GNN typically updates node features for a constant number of rounds, a recurrent GNN can perform updates for an unlimited number of rounds. Moreover, we consider two types of recurrent $\GNN$s; ones that use real numbers ($\GNN[\R]$s) and ones the use floating-point numbers ($\GNNF$s).

In this paper, we offer an alternative proof to the result in \cite{ahvonen_neurips} showing that $\GNN[\R]$s have---in restriction to monadic second order logic (MSO)---the same expressive power as the graded modal substitution calculus ($\GMSC$), a logic introduced in the mentioned article. The proof given in \cite{ahvonen_neurips} made use of parity tree automata (PTAs) and a related result that each MSO-sentence has a corresponding PTA over tree-shaped models.
The proof given here instead makes use of distributed message passing automata.

Graded modal substitution calculus \cite{ahvonen_neurips} is an extension of modal substitution calculus ($\msc$) with counting modalities. The logic $\msc$ was first introduced in \cite{Kuusisto13} where it was used to characterize finite message passing automata. In \cite{ahvonen_neurips} the logic $\GMSC$ was shown to characterize $\GNNF$s and a class of counting message passing automata ($\CMPA$s). Moreover, it was shown that, in restriction to node properties definable in MSO, $\GMSC$ also characterizes $\GNN[\R]$s; this is also proved in the present paper via different methods. The method centrally involves translating MSO-definable node properties into $\CMPA$s (or equivalently into $\GNNF$s) that are equivalent over trees to the MSO-formula.

We have corrected some typos and obvious glitches found in a previous version of this paper \cite{ahvonen2025classdistributedautomatacontains}; e.g., the $\mu$-calculus was accidentally called the $\mu$-fragment.

\section{Preliminaries}

We let $\N$, $\Z$, $\Z_+$ and $\R$ denote the sets of natural numbers, integers, positive integers and real numbers respectively. For all $k \in \Z_+$, we let $[k]$ denote the set $\{1, \dots, k\}$. Given a set $S$, we let $\cP(S)$ denote the power set of $S$, $\cP^+(S)$ denote $\cP(S) \setminus \{\emptyset\}$ and $\cM(S)$ denote the set of functions $S \to \N$, which are called \textbf{multisets}. We may also use double brackets to denote a multiset, i.e., we can write $\{\{s, s', s'\}\}$ to denote the multiset $M \colon \{s, s'\} \to \N$ where $M(s) = 1$ and $M(s') = 2$. Given $k \in \N$ and a multiset $M \colon S \to \N$, we let $M_{|k}$ denote the multiset $M'$ obtained from $M$ by defining for each $s \in S$ that $M_{|k}(s) = \min\{k, M(s)\}$.

Let $\prop = \{p_1, p_2, \dots\}$ be a countably infinite set of proposition symbols, and let $\var = \{X_1, X_2, \dots\}$ be a countably infinite set of schema variables. We usually write $\Pi$ to denote a finite subset of $\prop$ and $\cT$ to denote a finite subset of $\var$. A \textbf{Kripke model} over $\Pi$ (or $\Pi$-model) is a tuple $M = (W, R, V)$, where $W$ is a set of \textbf{nodes}, $R \subseteq W \times W$ is an \textbf{accessibility relation} and $V \colon W \to \cP(\Pi)$ is a valuation that assigns to each node a set of proposition symbols that are true in the node. A \textbf{pointed Kripke model} is a pair $(M,w)$, where $M$ is a Kripke model and $w$ is a node in $M$. If $(w,v) \in R$, then we say that $v$ is a \textbf{successor} or \textbf{neighbour} of $w$. A \textbf{node property} (over $\Pi$) is a class of pointed $\Pi$-models.

\subsection{Graph neural networks over reals and floats}

We introduce notions and notations related to recurrent graph neural networks in two frameworks: over real numbers and over floating-point numbers.

A \textbf{recurrent graph neural network} $\GNN[\R]$ over $(\Pi, d)$ is a tuple $\cG = (\R^d, \pi, \delta, F)$, where $\pi \colon \cP(\Pi) \to \R^d$ is an \textbf{initialization function}, $\delta \colon \R^d \times \cM(\R^d) \to \R^d$ is a \textbf{transition function} consisting of an aggregation function $\AGG \colon \cM(\R^d) \to \R^d$ and a combination function $\COM \colon \R^d \times \R^d \to \R^d$ such that
\[
    \delta(q, M) \colonequals \COM(q, \AGG(M)),
\]
and $F \subseteq \R^d$ is a set of \textbf{accepting feature vectors}.

A $\GNN[\R]$ over ($\Pi, d$) computes over a $\Pi$-model $M = (W, R, V)$ as follows. First, in round $0$, the initialization function assigns to each node $w \in W$ an initial feature vector $x_w^0 = \pi(V(w))$. In each subsequent round $t = 1, 2, \dots$, the feature vector of each node is updated as follows:
\[
    x_w^t = \COM(x_w^{t-1}, \AGG(\{\{ x_v^{t-1} \mid \text{$v$ is a successor of $w$} \}\})).
\]
A $\GNN[\R]$ \textbf{accepts} a node $w \in W$ if the node $w$ visits an accepting feature vector in any round of the run, i.e., $x_{w}^{t} \in F$ for some $t \in \N$.

An \textbf{R-simple $\GNN[\R]$} over $(\Pi, d)$ is a $\GNN[\R]$ over $(\Pi, d)$, where $\AGG$ is the element-wise sum and $\COM$ is defined as follows:
\[
\COM(q, p) = \mathrm{ReLU}^*( q C + p A + b),
\]
where $C , A \in \R^{d \times d} $, $b \in \R^d$ and $\mathrm{ReLU}^*(x) = \min(\max(0,x), 1)$ (applied separately to each vector element).

Let $p, q, \beta \in \Z_+$ and $\beta \geq 2$.
A \textbf{floating-point number over $p, q$ and $\beta$} is a pair
$(\pm \bs, \pm \be) = (\pm d_1 \cdots d_p, \pm e_1 \cdots e_q)$, where $\pm$ means the sign of a string is either $+$ or $-$ and $d_1, \ldots, d_p, e_1, \ldots, e_q \in [0; \beta-1]$.
Each such pair $(\pm d_1 \cdots d_p, \pm e_1 \cdots e_q)$ can be interpreted as a string $\pm 0.d_1 \cdots d_p \times \beta^{\pm e_1 \cdots e_q}$ that can be naturally interpreted as a real number. 
A \textbf{floating-point system over $p, q$ and $\beta$} is a tuple $(\cF, +_\cF, \cdot_\cF)$, where $\cF$ contains all the floating-point numbers over $p, q, \beta$, and $+_\cF$ and $\cdot_\cF$ are arithmetic operations of the form $\cF \times \cF \to \cF$. In more detail, $+_\cF$ is a saturating sum over $\cF$, meaning that first we take the precise sum of two floating-point numbers and then map the resulting number to the closest number in $\cF$. The operation $\cdot_\cF$ is defined analogously.

Now, $\GNNF$s are obtained from $\GNN[\R]$s by simply replacing reals with floating-point numbers. In the case of R-simple $\GNNF$s, the sum over reals and the multiplication over reals are replaced by the corresponding floating-point operations and the order of the element-wise sum is in the increasing order of floats. We make a natural assumption that $\GNNF$s are limited in their ability to distinguish neighbours, i.e., they have a bound $k$ such that $\AGG(M) = \AGG(M_{|k})$ for all multisets $M \in \cM(\cF^d)$.

\subsection{Logics}

In this section, we define the logics used in this paper. 

Graded modal logic ($\GML$) is obtained from modal logic by replacing the ordinary modality with a counting modality. More formally, let $\Pi$ be a set of proposition symbols. The \textbf{$\Pi$-formulae of $\GML$} are constructed according to the following grammar:
\[
    \varphi \coloncolonequals p_i \,|\, \neg\varphi \,|\, \varphi \lor \varphi \,|\, \Diamond_{\geq k} \varphi,
\]
where $p_i \in \Pi$ and $k \in \N$. We let $(M,w) \models \varphi$ denote that $\varphi$ is true in the pointed model $(M,w)$, and this is defined as usual for proposition symbols, negations and disjunctions. For formulae $\Diamond_{\geq k} \varphi$, we define that
\[
    (M,w) \models \Diamond_{\geq k} \varphi \text{ if and only if } (M,v) \models \varphi \text{ for at least $k$ successors $v$ of $w$.}
\]

The logic $\VGML$, as also studied in \cite{ahvonen_neurips}, consists of countably infinite disjunctions of $\GML$-formulae. More formally, \textbf{$\Pi$-formulae of $\VGML$} are constructed according to the following grammar:
\[
    \psi \coloncolonequals \varphi \,|\, \bigvee_{\varphi \in S} \varphi,
\]
where $\varphi$ is a $\Pi$-formula of $\GML$ and $S$ is an (at most) countably infinite set of $\Pi$-formulae of $\GML$.

Graded modal substitution calculus (or $\GMSC$), as studied in \cite{ahvonen_neurips} is a recursive extension of $\GML$. Let $\cT = \{X_1, \dots, X_k\} \subset \var$. The \textbf{$(\Pi, \cT)$-schemata of $\GMSC$} are constructed according to the following grammar:
\[
    \varphi \coloncolonequals p_i \,|\, X_i \,|\, \neg \varphi \,|\, \varphi \lor \varphi \,|\, \Diamond_{\geq k} \varphi,
\]
where $p_i \in \Pi$ and $X_i \in \cT$. A $(\Pi, \cT)$-program of $\GMSC$ consists of two lists
\[
    \begin{aligned}
        &X_{1}(0) \colonminus \varphi_1 \qquad &&X_1 \colonminus \psi_1 \\
        &\vdots &&\vdots \\
        &X_{n}(0) \colonminus \varphi_n &&X_n \colonminus \psi_n,
    \end{aligned}
\]
where $n \in \N$, each $\varphi_i$ is a $\Pi$-formula of $\GML$ and each $\psi_i$ is a $(\Pi, \cT)$-schema of $\GMSC$. Moreover, each program of $\GMSC$ is associated with a set $\cA \subseteq \cT$ of \textbf{appointed predicates}.

The semantics of $(\Pi, \cT)$-programs of $\GMSC$ is defined over pointed $\Pi$-models $(M,w)$, where $M = (W, R, V)$, as follows. We define that $(M,w) \models X_i^0$ if and only if $(M,w) \models \varphi_i$. Now, assume we have defined whether $(M,w) \models X_i^{t-1}$ for some $t \in \Z_+$ for all $i \in [n]$. For $(\Pi, \cT)$-schemata $\psi$ of $\GMSC$, we define whether $(M,w) \models \psi^t$ as follows. 
\begin{itemize}
    \item If $\psi \colonequals p_i \in \Pi$, then $(M,w) \models p_i^t$ if and only if $p_i \in V(w)$.
    \item If $\psi \colonequals \neg \varphi$, then $(M,w) \models (\neg \varphi)^t$ if and only if $(M,w) \not\models \varphi^t$.
    \item If $\psi \colonequals \varphi_1 \lor \varphi_2$, then $(M,w) \models (\varphi_1 \lor \varphi_2)^t$ if and only if $(M,w) \models \varphi_1^t$ or $(M,w) \models \varphi_2^t$.
    \item If $\psi \colonequals \Diamond_{\geq k} \varphi$, then $(M,w) \models (\Diamond_{\geq k} \varphi)^t$ if and only if $(M,v) \models \varphi^t$ for at least $k$ successors $v$ of $w$.
    \item If $\psi \colonequals X_i$, then $(M,w) \models X_i^t$ if and only if $(M,w) \models \psi_i^t$.
\end{itemize}
Finally, we define that the program \textbf{accepts} a pointed model $(M,w)$ if and only if $(M,w) \models X_j^t$ for some $t \in \N$ and some appointed predicate $X_j \in \cA$.

Monadic second-order logic (or $\MSO$) is a fragment of second-order logic (or SO) where only the quantification of unary second-order predicates is allowed.
Graded modal $\mu$-calculus (or $\cL_\mu$) \cite{DBLP:conf/cade/KupfermanSV02} is a fragment of $\MSO$ which extends graded modal logic by allowing formulae that calculate the least fixed point of a formula (or equivalently extend the modal $\mu$-calculus with the graded modality). This means constructing formulae of the type $\mu Z. \varphi$ whose semantics is defined in Kripke models such that the formula is true in the nodes belonging to the fixed-point of the function $f$ that maps each set $S$ of nodes (starting with the empty set $\emptyset$) to the set of nodes that satisfy $\varphi$ when the variable $Z$ is interpreted as $S$.

We say that a formula $\varphi$ of $\GML$, $\VGML$, the graded modal $\mu$-calculus or $\MSO$ \textbf{expresses} or \textbf{defines} a node property $P$ if $\varphi$ is true in exactly the pointed models in $P$. Likewise, we say that a program $\Lambda$ of $\GMSC$ \textbf{expresses} or \textbf{defines} a node property $P$ if $\Lambda$ accepts exactly the pointed models in $P$.

The \textbf{modal depth} of a formula of $\GML$, $\VGML$ or the graded modal $\mu$-calculus or a schema of $\GMSC$ is the maximum number of nested diamonds $\Diamond_{\geq k}$ in the formula or schema.

\subsection{Counting message-passing automata}

In this section, we give the definition and related notions for counting message passing automata.

Our definition for counting message passing automata is very similar to the one given in \cite{ahvonen_neurips} and based on message passing automata from \cite{Kuusisto13}. Given a finite set $\Pi$ of proposition symbols, a \textbf{counting message passing automaton} (or $\CMPA$) for $\Pi$ is a tuple $\cA = (Q, \pi, \delta, F)$, where 
\begin{itemize}
    \item $Q$ is a set of \textbf{states},
    \item $\pi \colon \cP(\Pi) \to Q$ is an \textbf{initialization function},
    \item $\delta \colon Q \times \cM(Q) \to Q$ is a \textbf{transition function} and
    \item $F \subseteq Q$ is a set of \textbf{accepting states}.
\end{itemize}
A \textbf{finite} $\CMPA$ (or $\FCMPA$) is a $\CMPA$ where the set of states is finite. We say that a $\CMPA$ is \textbf{bounded}, if there exists a $k \in \N$ such that for all states $q \in Q$ and all multisets $M \in \cM(Q)$ we have $\delta(q, M) = \delta(q, M_{|k})$. We call the smallest such $k$ the \textbf{bound} of $\cA$, and we say that $\cA$ is a $k$-$\CMPA$.

A $\CMPA$ $\cA = (Q, \pi, \delta, F)$ over $\Pi$ computes over a $\Pi$-model $M = (W, R, V)$ as follows. For each $t \in \N$, we define a \textbf{global configuration} $g_t \colon W \to Q$, which assigns a state to each node in $W$. In round $0$, we define that $g_0(w) \colonequals \pi(V(w))$ for each $w \in W$. Now assume we have defined $g_t$ for some $t \in \N$. We define that 
\[
    g_{t+1}(w) \colonequals \delta(g_t(w), \{\{ g_{t}(v) \mid \text{$v$ is a successor of $w$} \}\}).
\]
We say that $\cA$ \textbf{accepts} a pointed $\Pi$-model $(M,w)$ if $g_{t}(w) \in F$ for some $t \in \N$.

We say that a $\CMPA$ $\cA$ \textbf{expresses} or \textbf{defines} a node property $P$ if it accepts exactly the pointed models in $P$.

\subsection{Notions on equivalence}

Let $\cC[\Pi]$ be the class which contains all the $\CMPA$s over $\Pi$, all the $\GNN$s over $\Pi$, all the $\Pi$-programs of $\GMSC$ and all the $\Pi$-formulae of $\GML$, $\VGML$, the graded modal $\mu$-calculus and $\MSO$. We say that $x, y \in \cC[\Pi]$ are \textbf{equivalent} if they define exactly the same node property.

Let $\cA, \cB \subseteq \cC[\Pi]$. We say that $\cA$ and $\cB$ have the \textbf{same expressive power} if for each $x \in \cA$ there is an equivalent $y \in \cB$ and vice versa.

\section{Alternative proof of Theorem 4.1 in \cite{ahvonen_neurips}}\label{appendix: alternative MSO GNN reals and GMSC}

First, we recall Theorem 4.1 from \cite{ahvonen_neurips}.

\textbf{Theorem 4.1.}
\emph{Let $P$ be a property definable in $\MSO$. Then $P$ is expressible as a $\GNN[\R]$ if and only if it is definable in $\GMSC$.}

In this section, we provide an alternative proof for Theorem 4.1 in \cite{ahvonen_neurips}, showing that with respect to $\MSO$-properties, GNNs with reals are equally as expressive as $\GMSC$. (Note that the proof of Theorem 4.1 in \cite{ahvonen_neurips} is in the end of this section.) We begin by offering some necessary preliminary concepts relating to trees and automata. Then, we show in Proposition \ref{proposition: new prefix tree and extensions 2} that a formula $\varphi$ that is definable in $\MSO$ and $\VGML$ is true in a tree if and only if we can cut the tree off at some depth and then extend it from that depth arbitrarily without affecting the truth of the formula. Lemma \ref{MSO-definable properties} is an important tool in this section; it states that there exists an equivalent bounded $\FCMPA$ for each formula that is definable in $\MSO$ and $\VGML$. The proof of this lemma is split into multiple lemmas. The construction of an automaton with a fixed point accepting (and rejecting) condition is given according to subformulae of $\MSO$ in the proof of Theorem \ref{theorem: Buchi distributed MSO automaton} (the preceding technical lemmas \ref{properness automaton} and \ref{lemma: non-deterministic to deterministic} are used to deal with quantifiers). This is then transformed into an ordinary $\FCMPA$ by lemmas \ref{lemma: non-deterministic MSO automaton} and \ref{lemma: non-deterministic to deterministic} to prove Theorem 4.1 of \cite{ahvonen_neurips}.

Given a pointed Kripke model $M = (W, R, V)$ and $w, v \in W$, a \textbf{walk from $w$ to $v$} is a tuple $(v_{1}, \dots, v_{k+1})$ of nodes in $M$ such that $(v_{i}, v_{i+1}) \in R$ for each $i \in [k]$, $v_1 = w$ and $v_{k+1} = v$. Moreover, a \textbf{path from $w$ to $v$} is a walk $(v_1, \ldots, v_{k+1})$ from $w$ to $v$ such that for each $i, j \in [k+1]$ and $i \neq j$, we have $v_i \neq v_j$. The \textbf{length} of the walk (or resp., path) $(v_1, \ldots, v_{k+1})$ is $k$.

A \textbf{tree} is a (possibly infinite) Kripke model $T = (W, R, V)$ that contains exactly one node $w$ (called the \textbf{root}) such that $(v,w) \notin R$ for all $v \in W$, and where for each $v \in W$ there exists exactly one walk from $w$ to $v$ in $T$ (we call the length of this walk the \textbf{distance} from $w$ to $v$; recall that a node can occur repeatedly in a walk, but not in a path).
A \textbf{leaf} in a tree $T$ is a node $v$ that does not have any successors. A \textbf{rooted tree} is a pointed model $(T, w)$ where $T$ is a tree and $w$ is the root of $T$. 
The \textbf{depth} of a tree $T$ is the length of the longest path in $T$ (note that the longest path inevitably starts from the root and ends in a leaf). If there is no longest path in $T$, then the depth of $T$ is $\infty$.

Given a rooted tree $(T, w)$, we let $(T_k, w)$ denote the \textbf{$k$-prefix tree} of $(T, w)$, i.e., the restriction of $(T, w)$, where $v$ is in the domain of $T_{k}$ if and only if there is an at most $k$-long path from $w$ to $v$. This is consistent with $T_{k}$ as defined in \cite{ahvonen_neurips}. Similar to \cite{ahvonen_neurips}, we define that an \textbf{extension} of $(T_k, w)$ is then any $(T', w)$ whose $k$-prefix tree is $(T_{k}, w)$, that is, $(T', w)$ is obtained from $(T_k, w)$ by extending $(T_k, w)$ from the nodes at distance $k$ from $w$ by attaching subtrees, but not from any node at distance $\ell < k$ from $w$.
Likewise, a \textbf{$k$-extension} of a tree $(T, w)$ is any extension of $(T_{k}, w)$.
Note that if the depth of $(T, w)$ is $\ell < k$, then $(T, w) = (T_k, w) = (T', w)$ 
for all $k$-extensions $(T', w)$ of $(T, w)$.
Given an $\MSO$-formula $\varphi(x)$ (where the first-order variable $x$ is the only free variable), we say that a tree $(T, w)$ is \textbf{$k$-extendable w.r.t $\varphi(x)$} if
for \emph{all} $k$-extensions $(T', w)$ of $(T, w)$, we have $T' \models \varphi(w)$.

The following proposition shows that for any $\MSO$-formula $\varphi(x)$ that is equivalent to a formula of $\VGML$, the formula $\varphi(x)$ being true in a tree $T$ implies that $\varphi(x)$ is true in all $k$-extensions of $T$ for some $k \in \N$. 
\begin{proposition}\label{proposition: new prefix tree and extensions 2}
    Let $\varphi(x)$ be an $\MSO$-formula definable in $\VGML$. Then for all rooted trees $(T,w)$: if $T \models \varphi(w)$, there is a $k \in \N$ such that $(T,w)$ is $k$-extendable w.r.t. $\varphi(x)$.
\end{proposition}
\begin{proof}
    Let $(T, w)$ be a rooted tree. Let $\psi$ be an $\VGML$-formula that defines $\varphi(x)$. Assume that $T \models \varphi(w)$. Thus $T, w \models \psi$ i.e., $T, w \models \psi'$ for some disjunct $\psi'$ of $\psi$. Choose $k$, where $k$ is the modal depth of $\psi'$. 
    Let $(T', w)$ be an extension of $(T_k, w)$.
    Now, $T, w \models \psi'$ implies $T', w \models \psi'$ and thus we have $T', w \models \psi$ and ultimately $T' \models \varphi(w)$, since the modal depth of $\psi'$ is $k$. Therefore, $(T, w)$ is $k$-extendable w.r.t. $\varphi$.
\end{proof}

The main lemma of this section is proven in the sections below and is stated as follows. 
\begin{lemma}\label{MSO-definable properties}
    Given an $\MSO$-formula $\varphi(x)$ definable by an $\VGML$-formula, we can construct a bounded carefree $\FCMPA$ $\cA$ such that for all rooted trees $(T, w)$ we have
    \[
    \text{$(T, w)$ is accepted by $\cA$} \iff T \models \varphi(w).
    \]
\end{lemma}

\subsection{An informal strategy}

We begin by defining some useful tools: 
\begin{itemize}
    \item non-deterministic automata, which modify ordinary automata by giving each node a set of possible initial states to choose from instead of just one,
    \item automata with rejecting states, which means that the automaton can explicitly reject nodes instead of just not accepting them,
    \item quasi-acyclic automata, meaning that in each round, the next state of each node is either the same state or one that the node has not previously visited,
    \item forgetful automata, meaning that state transitions in each node do not depend on the current state of that node, and
    \item fixed point accepting conditions, which means that a node is accepted if it eventually reaches an accepting state and stays in the same state in all subsequent rounds.
\end{itemize}
Then we proceed by constructing the automaton mentioned in Lemma \ref{MSO-definable properties} as follows. Considering an $\MSO$-formula $\varphi$ definable by an $\VGML$-formula, we do the following:
\begin{itemize}
\item We construct an equivalent bounded $\FCMPA$ in Theorem \ref{theorem: Buchi distributed MSO automaton} that is also quasi-acyclic  and forgetful, with fixed point conditions for accepting and rejecting pointed models. The automaton is built via induction over the structure of $\varphi$.
    \begin{itemize}
    \item We start by constructing automata for atomic formulas. Negations and conjunctions are handled by simple manipulations of the automata obtained in previous steps.
    \item Existential quantifiers are handled in two steps.
    \begin{enumerate}
        \item For first-order existential quantifiers, we pre-emptively combine the automaton received in previous steps with one that checks that the quantified first-order variable $y$ is assigned to precisely one node in the tree (this automaton is constructed in Lemma \ref{properness automaton}; we skip this step in the case of second-order existential quantifiers).
        \item We turn the automaton into one where each node guesses in each round whether the quantified variable is true or not. The automaton keeps track of all possible outcomes. This utilizes the well-known power set construction, which is also used in Lemma \ref{lemma: non-deterministic to deterministic}.
    \end{enumerate}
    \end{itemize}
\item Finally, the obtained automaton with fixed point conditions for accepting and rejecting is transformed into an equivalent non-deterministic $\FCMPA$ with an omnipresent accepting condition (meaning that every run of the automaton accepts in the same round) in Lemma \ref{lemma: non-deterministic MSO automaton}, which we then transform into an equivalent deterministic $\FCMPA$ with just the ordinary accepting condition using Lemma \ref{lemma: non-deterministic to deterministic}. This is the final step in the proof of Lemma \ref{MSO-definable properties}.
\end{itemize}

We will use a modification of $\CMPA$s to construct the automaton. Intuitively, a modified $\CMPA$ has a unique transition function for each set $P \subseteq \Pi$ of proposition symbols. For the rest of the computation the automaton uses in each node the transition function whose label corresponds to the set of proposition symbols that are true in that node. More formally, such a $\CMPA$ is a tuple $\cA = (Q, \pi, (\delta_P)_{P \subseteq \Pi}, F)$ over the vocabulary $\Pi$. The computation of $\cA$ is defined in Kripke models $M$ over $\Pi$ as follows. The initial step is analogous to ordinary $\CMPA$s. For the rest of the computation, in each node $w$ in $M$, the automaton $\cA$ uses the transition function $\delta_P \colon Q \times \cM(Q) \to Q$ analogously to ordinary $\CMPA$s, where $P$ is the set of proposition symbols that are true in $w$. Note that modified $\CMPA$s and ordinary $\CMPA$s have the same expressive power (and the same holds for other subclasses, e.g., $\FCMPA$s), so there is no true distinction between the definitions. For the rest of this section, $\CMPA$ will always refer to a modified one (the same holds for other subclasses of $\CMPA$s as well). During this section we might call $\CMPA$s simply automata, since we do not consider other types of automata in this section.

\subsection{Notions on automata}

In order to prove Lemma \ref{MSO-definable properties}, we need additional modifications and properties of automata that arise naturally in our construction. The two modifications are non-deterministic automata and automata with rejecting states. The two central properties of automata are quasi-acyclicity, which means that no loops are allowed in transition diagrams apart from self-loops, and forgetfulness, which means that a node's transitions do not depend on its own previous state.

A \textbf{non-deterministic $\CMPA$} for $\Pi$ is a tuple $\cA = (Q, \pi, (\delta_P)_{P \subseteq \Pi}, F)$, where $Q$, $\delta_{P}$ and $F$ are all defined analogously to deterministic $\CMPA$s and $\pi \colon \cP(\Pi) \to \cP(Q) \setminus \{\emptyset\}$. A non-deterministic automaton runs like a deterministic one, but in the initialization round the initialization function $\pi$ gives a set of states, whence the automaton chooses a state to proceed into. \textbf{Global configurations} are defined as before, except that for $g_{0}$ we define that $g_{0}(v) \in \pi(P)$ where $P$ is the set of proposition symbols true in $v$. (Notice that different choices in initial configurations define potentially different global configurations in subsequent rounds.) Given a pointed Kripke model $(M, w)$ over $\Pi$, a \textbf{run of $\cA$ in $(M, w)$} is a sequence $(g_{n}(w))_{n \in \N}$ (where $g_{0}(v) \in \pi(V(v))$ for each $v \in W$).

An \textbf{accepting run of $\cA$ in $(M, w)$} is a run of $\cA$ in $(M, w)$ such that $g_{n}(w) \in F$ for some $n \in \N$. A non-deterministic automaton \textbf{accepts} a pointed Kripke model $(M, w)$ if there exists an accepting run.

We also define a stronger definition for accepting where each run accepts simultaneously; we say that a non-deterministic automaton 
\textbf{$k$-omnipresently accepts} a pointed Kripke model $(M,w)$, if there is a round $k \in \N$ such that every run of the automaton is in an accepting state at $w$ in round $k$, and $k$-omnipresent rejecting is defined analogously.
If it is clear from the context, a deterministic $\CMPA$ is simply referred to as a $\CMPA$. We will make a clear distinction when discussing non-deterministic $\CMPA$s. The notion of non-determinism extends for other variants of $\CMPA$s such as $\FCMPA$s.

Next, we define a variant of automata that can also reject pointed models. A $\CMPA$ (deterministic or not) \textbf{with rejecting states} is a tuple $(Q, \pi, (\delta_P)_{P \subseteq \Pi}, F, F')$, where $F' \subseteq Q \setminus F$ is a set of \textbf{rejecting states} and $(Q, \pi, (\delta_P)_{P \subseteq \Pi}, F)$ is a $\CMPA$. Moreover, acceptance is defined analogously to an automaton which includes only accepting states. Such an automaton \textbf{rejects} a pointed Kripke model $(M, w)$ if \textbf{(1)} it is deterministic and $(M, w)$ visits a rejecting state and does not visit an accepting state, or \textbf{(2)} it is non-deterministic and there exists at least one rejecting run and no accepting run in $(M, w)$.\footnote{Note that \emph{rejecting} and \emph{not accepting} (resp., \emph{accepting} and \emph{not rejecting}) are two distinct phenomena; even if an automaton does not accept a pointed model, it does not follow that it necessarily rejects the pointed model, although rejecting does imply not accepting.}

When we say that a deterministic automaton \textbf{accepts} or \textbf{rejects $(M,w)$ in round $i$}, we mean that the automaton is in an accepting state or a rejecting state respectively at $w$ in round $i$.
We say that a non-deterministic automaton \textbf{accepts $(M,w)$ in round $i$} if there is a run that is in an accepting state at $w$ in round $i$, and it \textbf{rejects $(M,w)$ in round $i$} if there is a run that is in a rejecting state at $w$ in round $i$ and no run that is in an accepting state at $w$ in round $i$.

We say that a deterministic automaton \textbf{fixed point accepts} (resp., \textbf{fixed point rejects}) a pointed model $(M, w)$ if there exists a $k \in \N$ such that the state of $w$ in round $k$ is the same as in round $k'$ for all $k' > k$, and the state of $w$ in round $k$ is an accepting state (resp., a rejecting state). 
Respectively, a non-deterministic automaton \textbf{fixed point accepts} (or resp. \textbf{fixed point rejects}) a pointed model $(M, w)$ if there exists a run such that $w$ is fixed point accepted (or resp., there is a run where $w$ is fixed point rejected and no run where $w$ is fixed point accepted).

Next, we define quasi-acyclicity and forgetfulness. Informally, an automaton is quasi-acyclic if its transition functions do not form any cycles (except self-loops are allowed), and it is forgetful if it cannot remember a node's previous state. More formally, an automaton (deterministic or not) is \textbf{quasi-acyclic} if for each sequence of states $g_{i}(w), \ldots, g_{i+k}(w)$ during a run of the automaton in any pointed Kripke model $(M,w)$, we have that if $g_{i}(w) = g_{i+k}(w)$, then $g_{i+j-1}(w) = g_{i+j}(w)$ for all $j \in [k]$. An automaton (deterministic or not) is said to be \textbf{forgetful} if for all $P \subseteq \Pi$ and all $M \in \cM(Q)$, we have $\delta_P(q', M) = \delta_P(q'', M)$ for all states $q', q'' \in Q$. In other words the transition function $\delta_P$ is simply a function of the type $\cM(Q) \to Q$ (note that the initial state of a node is still remembered by the choice of $\delta_{P}$). An automaton that is both quasi-acyclic and forgetful is called \textbf{carefree}.

\subsection{Auxiliary results}

In this section, we prove two auxiliary lemmas. The first one shows how to construct an automaton that checks whether first-order variables are assigned correctly in a Kripke model. The second one shows how to translate a non-deterministic forgetful automaton into a deterministic one, preserving quasi-acyclicity.

Given a set $\cX$ of first-order and monadic second-order variables and a rooted tree $(T, w)$, an \textbf{interpretation} of $(T, w)$ over $\cX$ is a rooted tree $(T^*, w)$, such that each first-order variable in $\cX$ is interpreted as a singleton in $(T^*, w)$ and each monadic second-order variable in $\cX$ is interpreted in an arbitrary way in $(T^*, w)$. 
More formally, if $T = (W, R, V)$ is a $\Pi$-model, then $T^{*} = (W, R, V^{*})$ is any $\Pi \cup \Pi_{\cX}$-model where $\Pi_{\cX} = \{\, p_{x} \mid x \in \cX \,\}$. The valuation can be defined in any way such that $V^*(w) \cap (\Pi \setminus \Pi_{\cX}) = V(w)$ and $\abs{\{\, v \in W \mid p_{x} \in V^{*}(v) \,\}} = 1$ for all \emph{first-order} variables $x \in \cX$.
An \textbf{unrestricted interpretation} of $(T, w)$ over $\cX$ is otherwise the same as an interpretation, but we relax the condition for first-order variables.

The following lemma shows that we can construct an automaton to check that an unrestricted interpretation of a rooted tree is an interpretation in the regular sense, i.e., that each first-order variable is interpreted as a singleton.
\begin{lemma}\label{properness automaton}
    Given a set $\cX$ of first order variables, we can construct a carefree $2$-$\FCMPA^{r}$ $\cA$ such that for all rooted trees $(T, w)$ and all unrestricted interpretations $(T^*,w)$ over $\cX$:
    \begin{enumerate}
        \item $\cA$ accepts $(T^*,w)$ in round $i$ if and only if each variable $x \in \cX$ is interpreted as a singleton in $(T^*_{i}, w)$,
        \item $\cA$ rejects $(T^*, w)$ in round $i$ if and only if any variable $x \in \cX$ is interpreted as a non-empty set in $(T^*_i, w)$ that is not a singleton. 
    \end{enumerate}
\end{lemma}
\begin{proof}
    Intuitively, in step $0$ each node makes a note that each variable in its signature has been seen. In step $i$, a node sees how many of each variable its out-neighbours have seen at depth $i-1$; all these numbers are independent of each other, because the branches starting from each out-neighbour are pairwise disjoint. Thus, each node can combine its initial state with how many of each (first-order) variable its out-neighbours have seen at depth $i-1$, and conclude 1) whether each variable has been seen within its $i$-depth neighbourhood and 2) whether any variable has been seen more than once within its $i$-depth neighbourhood. The automaton accepts if each variable has been seen exactly once and rejects if at least one variable has been seen more than once.
    
    More formally, let $\cX = \{x_{1}, \dots, x_{n}\}$ be the set of first-order variables. We construct an automaton $\cA_{\text{proper }\cX} = (Q, \pi, (\delta_{P})_{P \subseteq \Pi}, F, F')$ over $\Pi \cup \Pi_{\cX}$ as follows.
    \begin{itemize}
        \item The set of states is $Q = \{\, q_{x_{1} \geq k_{1}, \dots, x_{n} \geq k_{n}} \mid k_{1}, \dots, k_{n} \in \{0, 1, 2\} \,\}$, i.e., the state of a node $w$ at time $i$ keeps track of how many times each variable has been encountered in the $i$-prefix tree $T_{i}$ generated by $w$; $0$, $1$, or at least $2$ times. For example, if $n = 3$ then $q_{x_{1} \geq 0, \, x_{2} \geq 2, \, x_{3} \geq 1}$ says that $x_{1}$ has not been encountered, $x_{2}$ has been encountered at least twice and $x_{3}$ has been encountered just once.
        \item The initialization function is defined as follows. For all $P \subseteq \Pi \cup \Pi_{\cX}$, $\pi(P) = q_{x_{1} \geq k_{1}, \dots, x_{n} \geq k_{n}}$ where $k_{i} = 1$ iff $p_{x_{i}} \in P$ and $k_{i} = 0$ iff $p_{x_{i}} \notin P$ for all $i \in [n]$.
        \item The transition functions are defined as follows. Let $\pi(P) = q_{x_{1} \geq k_{1}, \dots, x_{n} \geq k_{n}}$ and let
        \[
            M = \{\{ q_{x_{1} \geq k^{1}_{1}, \dots, x_{n} \geq k^{1}_{n}}, \dots, q_{x_{1} \geq k^{m}_{1}, \dots, x_{n} \geq k^{m}_{n}} \}\} \in \cM(Q)
        \]
        be a multiset of $m$ states. Then we define 
        \[
        \delta_{P}(M) = q_{x_{1} \geq \ell_{1}, \dots, x_{n} \geq \ell_{n}}, \text{ where } \ell_{i} = \min\{k_{i} + \sum_{j = 1}^{m} k^{j}_{i}, 2\}.
        \]
        \item The set of accepting states is the singleton $F = \{q_{x_{1} \geq 1, \dots, x_{n} \geq 1}\}$ and the set of rejecting states is $F' = \{\, q_{x_{1} \geq m_1, \dots, x_{n} \geq m_n} \mid (m_1, \ldots, m_n) \in \{0, 1, 2\}^n \setminus \{0,1\}^n \,\}$. 
    \end{itemize}
    It is easy to show by induction over the run of the automaton that the following holds for all rooted trees $(T, w)$ and all unrestricted interpretations $(T^*,w)$ over $\cX$:
    \begin{enumerate}
        \item $\cA$ accepts $(T^*,w)$ in round $i$ if and only if in $(T^*_{i}, w)$ each variable $x \in \cX$ is interpreted as a singleton,
        \item $\cA$ rejects $(T^*, w)$ in round $i$ if and only if in $(T^*_i, w)$ any variable $x \in \cX$ is interpreted as a non-empty set that is not a singleton. 
    \end{enumerate}
    The automaton is clearly carefree by its construction.
\end{proof}

In the following lemma, we show how given a non-deterministic automaton with rejecting states, we can construct an equivalent deterministic automaton with rejecting states.
\begin{lemma}\label{lemma: non-deterministic to deterministic}
    Given a non-deterministic forgetful bounded $\FCMPA^r$ $\cA$, we can construct a deterministic forgetful bounded $\FCMPA^r$ $\cA'$ such that for all rooted trees $(T, w)$ and all $i \in \N$:
    \begin{enumerate}
        \item $\cA$ accepts (resp. fixed point accepts) $(T, w)$ iff $\cA'$ accepts (resp. fixed point accepts) $(T, w)$.
        \item $\cA$ rejects (resp. fixed point rejects) $(T, w)$ iff $\cA'$ rejects (resp. fixed point rejects) $(T, w)$.
    \end{enumerate}
    If $\cA$ is quasi-acyclic, then $\cA'$ is also quasi-acyclic. Moreover, $\cA'$ can be modified such that $\cA$ $k$-omnipresently accepts (resp. $k$-omnipresently rejects) $(T, w)$ iff $\cA'$ accepts (resp. rejects) $(T, w)$.
\end{lemma}
\begin{proof}
    We use the power set construction to determinize a non-deterministic automaton.
    Assume we have a non-deterministic forgetful bounded $\FCMPA$ with rejecting states, i.e., an automaton $\cA = (Q, \pi, (\delta_{P})_{P \subseteq \Pi}, F_{1}, F_{2})$ such that $\pi \colon \cP(\Pi) \to \cP^+(Q)$, $\delta_{P} \colon \cM(Q) \to Q$ for each $P \subseteq \Pi$ and $F_{1}, F_{2} \subseteq Q, F_{1} \cap F_{2} = \emptyset$.
    Intuitively, the deterministic automaton keeps track of every run of $\cA$ simultaneously; the state of a node tells its state in each run, and the accepting and rejecting states are defined according to the accepting and rejecting conditions for non-deterministic automata.
    More formally, we construct the equivalent deterministic forgetful bounded $\FCMPA^{r}$ $\cA' = (Q', \pi', (\delta_{P}')_{P \subseteq \Pi}, F_{1}', F_{2}')$ as follows. The set of states is $Q' = \cP(Q)$ and the initialization function is $\pi' = \pi$. We define $\delta_{P}' \colon \cM(Q') \to Q'$ as follows. Let $M = \{\{Q_{1}, \dots, Q_{m}\}\} \in \cM(Q')$. We define that
    \[
        \delta_{P}'(M) = \{\, q' \in Q \mid \delta_{P}(\{\{q_{1}, \dots, q_{m}\}\}) = q' \text{ where }  q_{i} \in Q_{i} \text{ for all } i \in [m] \,\}.
    \]
    Finally we set that $F_{1}' = \{\, q \in Q' \mid q \cap F_{1} \neq \emptyset \,\}$ and $F_{2}' = \{\, q \in Q' \mid q \cap F_{1} = \emptyset, q \cap F_{2} \neq \emptyset \,\}$.
    For the case with $k$-omnipresent acceptance, we define instead that $F_{1}' = \{\, q \in Q' \mid q \subseteq F_{1} \,\}$ and $F_{2}' = \{\, q \in Q' \mid q \subseteq F_{2} \,\}$.

    First, we show that $\cA'$ is bounded. Consider for example the case where the bound of $\cA$ is $k = 1$. In the worst possible situation, a node might have at least $\abs{Q}$ neighbours such that the set of possible states for each neighbour is $Q$. Now, it is possible that each of the $\abs{Q}$ neighbours picks a different state, in which case $\cA$ could distinguish between all of them. Therefore, $\cA'$ must be able to distinguish between $\abs{Q}$ identical nodes. In general, for any $k$ a node might have $k \abs{Q}$ neighbours with the set $Q$ of possible states, and for each state in $Q$ there might be at least $k$ neighbours with that state, meaning that $\cA'$ must be able to distinguish between $k \abs{Q}$ identical nodes. In other words, if $\cA$ is a $k$-$\FCMPA^r$, then $\cA'$ is a $k \abs{Q}$-$\FCMPA^r$.

    It is easy to prove the following by induction over the run of an automaton for all rooted trees $(T,w)$.
    \begin{itemize}
        \item Let $(Q_n)_{n \in \N}$ be the run of $\cA'$ in $(T, w)$. Let $\cR$ denote the set of all possible runs of $\cA$ in $(T, w)$ and let $\cR'$ denote the set of all sequences $(q_n)_{n \in \N}$, where $q_n \in Q_n$. Then $\cR = \cR'$. 
    \end{itemize}
    In the induction step the key property is the forgetfulness of $\cA$. Since $\cA$ is forgetful, a node's next state only depends on the possible states of its neighbours and therefore each combination of the neighbours' possible states is possible. Otherwise, the automaton could end up in a state which is not even ``reachable'' (in the transition diagram of the transition function). 
    
    Now, with the result above, we prove cases 1 and 2 of the lemma. Basic accepting and rejecting are simple, as is the case for $k$-omnipresent accepting and rejecting, so we go over the 
    fixed point conditions.
    \begin{enumerate}
        \item The automaton $\cA$ fixed point accepts $(T,w)$ if and only if there is a run of $\cA$ at $(T,w)$ such that $\cA$ ends up in a fixed point state from the set $F_1$. Since $\cA$ is carefree and by the result above, this is equivalent to $\cA'$ ending up in a fixed point state from the set $F'_1$. This is true if and only if $\cA'$ fixed point accepts $(T,w)$.
        \item The automaton $\cA$ fixed point rejects $(T,w)$ if and only if there is a run of $\cA$ at $(T,w)$ such that $\cA$ ends up in a fixed point state from the set $F_2$, and there is no run of $\cA$ at $(T,w)$ such that $\cA$ ends up in a fixed point state from $F_1$. Since $\cA$ is carefree and by the result above, this is equivalent to $\cA'$ ending up in a fixed point state from $F_{2}'$; thus if $\cA'$ ended up in a fixed point state from $F_{1}'$, there would have to be a fixed point accepting run of $\cA$ at $(T,w)$. In any case, this is true if and only if $\cA'$ fixed point rejects $(T,w)$.
    \end{enumerate}

    Next we show that if $\cA$ is quasi-acyclic then $\cA'$ is also. We first show that if $(Q_{0}, \dots, Q_{t})$ is the state transition diagram for $\cA'$ at $(T, w)$ in round $t$ (i.e., $\cA'$ is in state $Q_{i}$ in round $i$ at $w$ for all $i \in [0;t]$), then each $(q_{0}, \dots, q_{t})$ is a possible state transition diagram for $\cA$ at $(T, w)$ in round $t$ (i.e., $\cA$ is in state $q_{i}$ in round $i$ at $w$ for all $i \in [0;t]$), where $q_{i} \in Q_{i}$ for all $i \in [0;t]$.
    
    We prove the claim by induction over $t$. For $t = 0$, we have that $Q_{0} = \pi(P)$ for some $P \subseteq \Pi$, and by definition each $q_{0} \in Q_{0}$ is a possible state in round $0$. Now assume the claim holds for $t$, and $(Q_{0}, \dots, Q_{t+1})$ is the state transition diagram for $\cA'$ at $(T, w)$. Each neighbour $w'$ of $w$ has a transition diagram $(Q_{0}', \dots, Q_{t}')$, and by the induction hypothesis each $(q_{0}', \dots, q_{t}')$ is a possible state transition diagram for $\cA$ at $(T, w')$, where $q_{i}' \in Q_{i}'$ for all $i \in [0;t]$. In particular, each $q_{t}'$ is a valid continuation of $(q_{0}', \dots, q_{t-1}')$, and each $q_{t}$ corresponds to at least some combination of state transition diagrams $(q_{0}', \dots, q_{t-1}')$ from the neighbours of $w$.
    By the definition of $\delta_{P}'$, every $q_{t+1} \in Q_{t+1}$ is obtained by applying $\delta_{P}$ to states $q_{t}' \in Q_{t}'$, and is a valid continuation to each $q_{t} \in Q_{t}$ as $q_{t}$ is independent of each $q_{t}'$. This completes the induction step.

    Now it is easy to see that $\cA'$ is quasi-acyclic (if $\cA$ is also). If there was a loop $(Q_{1}, \dots, Q_{n})$ in the state transition diagram of $\cA'$ at $(T, w)$ such that $Q_{1} = Q_{n}$ and $Q_{i} \neq Q_{1}$ for some $i \in [n]$, then by the above induction every $(q_{1}, \dots, q_{n})$ would be a possible state transition diagram of $\cA$ at $(T, w)$. In particular, since $Q_{i} \neq Q_{1}$, we can choose some $q_{1} \in Q_{1} \setminus Q_{i}$ or $q_{i} \in Q_{i} \setminus Q_{1}$. In either case we choose that $q_{1} = q_{n}$. Now $(q_{1}, \dots, q_{i}, \dots, q_{n})$ is a loop in the state transition diagram of $\cA$ at $(T, w)$ such that $q_{1} = q_{n}$ and $q_{i} \neq q_{1}$ for some $i \in [n]$, which is a contradiction.
\end{proof}

\subsection{Expressing MSO with distributed automata}

Now, we are ready to prove the following important lemma.
\begin{theorem}\label{theorem: Buchi distributed MSO automaton}
    Given an $\MSO$-formula $\varphi(x)$ 
    we can construct a deterministic carefree bounded $\FCMPA^{r}$ $\cA_\varphi$ such that for all rooted trees $(T, w)$ we have
    \[
    T \models \varphi(w) \iff \text{$\cA_\varphi$ fixed point accepts $(T, w)$}
    \]
    and
    \[
    T \not\models \varphi(w) \iff \text{$\cA_\varphi$ fixed point rejects $(T, w)$}.
    \] 
\end{theorem}
\begin{proof}
Intuitively, we construct an automaton that computes the truth of an MSO-formula in a rooted tree; the truth value might be determined at the root, or it may be determined in some branch of the tree, from where it is then passed down to the root. The truth of the formula may be incorrectly determined at first, but is eventually corrected, hence the 
fixed point accepting condition. 

Let $\cX = \{x_1, \ldots, x_\ell, X_1, \ldots, X_k \}$ be a set of first-order variables $x_1, \ldots, x_\ell$ and monadic second-order variables $X_1, \ldots, X_k$.
We prove by induction over the structure of a formula $\varphi(x_1, \ldots, x_\ell, X_1, \ldots, X_k)$ of $\MSO$ over $\Pi$ that there exists a deterministic 
forgetful quasi-acyclic bounded $\FCMPA^{r}$ $\cA_{\varphi}$ over $\Pi \cup \Pi_\cX$ such that
for all rooted trees $(T, w)$ over $\Pi$ and for all interpretations $(T^*, w)$ of $(T, w)$ over $\cX$ we have 

\begin{enumerate}
    \item $T^* \models \varphi \iff$ $(T^*, w)$ is fixed point accepted by $\cA_\varphi$,
    \item $T^* \not\models \varphi \iff$ $(T^*, w)$ is fixed point rejected by $\cA_\varphi$.
\end{enumerate}

Base cases:
\begin{enumerate}
    \item
    Case $\varphi \colonequals Py$. The automaton $\cA_{Py}$ uses two states: $q_{Py}$ and $q_{\neg Py}$. The state $q_{Py}$ is the only accepting state and $q_{\neg Py}$ is the only rejecting state.

    In the initial step, $\cA_{Py}$ checks if a node is labeled with $P$ and $y$ (or more specifically, $p_P$ and $p_y$) and goes to the state $q_{Py}$ if that is the case. Otherwise, the node enters the state $q_{\neg Py}$.

    In the transition step, the automaton $\cA_{Py}$ does the following in each node $w$. If the initial state of $w$ is $q_{Py}$, or if $w$ receives $q_{Py}$ from at least one of its neighbours, then $\cA_{Py}$ enters into the state $q_{Py}$. If the initial state of $w$ is $q_{\neg Py}$ and $w$ receives $q_{\neg Py}$ from all of its neighbours, then $\cA_{Py}$ enters into the state $q_{\neg Py}$.

    Clearly, for all rooted trees $(T, w)$ over $\Pi$ and all interpretations $(T^*, w)$ of $(T, w)$ over $\{y\}$, $\cA_{Py}$ fixed point accepts $(T^*, w)$ if and only if $T^* \models Py$ and $\cA_{Py}$ fixed point rejects $(T^*, w)$ if and only if $T^* \not\models Py$. 
    Also, $\cA_{Py}$ is clearly bounded (and the bound is $1$).
    \item 
    Case $\varphi \colonequals Ryz$. The automaton $\cA_{Ryz}$ uses four states: $q_{Ryz}$, $q_{\neg Ryz}$, $q_y$ and $q_z$. The state $q_{Ryz}$ is the only accepting state and the others are rejecting states.

    In the initial step (i.e. in the zeroth round), the automaton $\cA_{Ryz}$ checks if a node is labeled with $y$ and $z$. If it is labeled with $z$ but not $y$, then the automaton enters the state $q_z$. If it is labeled with $y$ but not $z$, the automaton enters the state $q_y$. Otherwise, the automaton enters the state $q_{\neg Ryz}$.

    In the transition step, the automaton $\cA_{Ryz}$ does the following in each node $w$. If the initial state of $w$ is $q_z$, then the automaton stays in the state $q_z$. If $w$ has the initial state $q_y$ and it does not receive a message $q_z$, then the automaton remains in the state $q_y$. If the initial state of $w$ is $q_y$ and it receives the message $q_z$, or if the initial state of $w$ is $q_{\neg Ryz}$ and it receives the message $q_{Ryz}$, then the automaton enters the state $q_{Ryz}$. Otherwise, the automaton enters the state $q_{\neg Ryz}$.

    Clearly, for all rooted trees $(T,w)$ over $\Pi$ and all interpretations $(T^*, w)$ of $(T, w)$ over $\{y, z\}$, $\cA_{Ryz}$ fized point accepts $(T^*, w)$ if and only if $T^* \models Ryz$ and $\cA_{Ryz}$ fixed point rejects $(T^*, w)$ if and only if $T^* \not\models Ryz$. 
    Also, $\cA_{Ryz}$ is clearly bounded (and the bound is $1$).
    \item 
    Case $\varphi \colonequals y = z$. The automaton $\cA_{y=z}$ uses two states: $q_{y = z}$ and $q_{y \neq z}$. The state $q_{y = z}$ is the only accepting state and $q_{y \neq z}$ is the only rejecting state.

    In the initial round, the automaton $\cA_{y=z}$ enters the state $q_{y = z}$ if the node is labeled by $y$ and $z$. Otherwise, it enters the state $q_{y \neq z}$.

    In the transition rounds, $\cA_{y = z}$ performs as follows in each node $w$. If the initial state of $w$ is $q_{y = z}$, or if it receives the message $q_{y = z}$, the automaton enters the state $q_{y = z}$. Otherwise, the automaton enters the state $q_{y \neq z}$.

    Clearly, for all rooted trees $(T,w)$ over $\Pi$ and all interpretations $(T^*, w)$ of $(T, w)$ over $\{y, z\}$, $\cA_{y = z}$ fixed point accepts $(T^*, w)$ if and only if $T^* \models y=z$ and $\cA_{y = z}$ fixed point rejects $(T^*, w)$ if and only if $T^* \not\models y = z$. 
    Also, $\cA_{y = z}$ is clearly bounded (and the bound is $1$).
\end{enumerate}

\textbf{Induction hypothesis:} 
Assume that the induction claim holds for $\MSO$-formulae $\psi$, $\psi_1$ and $\psi_2$. Let $\cY$, $\cY_1$ and $\cY_2$ denote the sets of free variables in these formulae respectively. We let $\cA_\psi$ denote the corresponding automaton for $\psi$ over $\Pi \cup \Pi_\cY$. Analogously, we let $\cA_{\psi_1}$ and $\cA_{\psi_2}$ denote the corresponding automata for $\psi_1$ and $\psi_2$ respectively. Note that, by the induction hypothesis, these automata are deterministic, forgetful and quasi-acyclic.
\begin{enumerate}
    \item Case $\varphi \colonequals \neg \psi$.
    The automaton $\cA_{\neg \psi}$ is obtained from $\cA_{\psi}$ by designating that accepting states of $\cA_{\psi}$ are rejecting states of $\cA_{\neg \psi}$ and rejecting states of $\cA_{\psi}$ are accepting states of $\cA_{\neg \psi}$.

    Now, for all rooted trees $(T,w)$ and all interpretations $(T^*, w)$ of $(T, w)$ over $\cY$, $T^* \models \neg \psi$ iff $T^* \not\models \psi$ iff $\cA_\psi$ fixed point rejects $(T^*, w)$ (by the induction hypothesis) iff $\cA_{\neg \psi}$ fixed point accepts $(T^*, w)$.
    Similarly, $T^* \not\models \neg \psi$ iff $T^* \models \psi$ iff $\cA_\psi$ fixed point accepts $(T^*, w)$ (by the induction hypothesis) iff $\cA_{\neg \psi}$ fixed point rejects $(T^*, w)$. It is also easy to see that the automaton is forgetful and quasi-acyclic
    since we only flipped the accepting states with the rejecting states. 
    Moreover, it is clear that $\cA_{\neg \psi}$ has the same bound as $\cA_{\psi}$.
    \item Case $\varphi \colonequals \psi_1 \land \psi_2$. 
    Let $\Pi_1 = \Pi \cup \Pi_{\cY_1}$ and $\Pi_2 = \Pi \cup \Pi_{\cY_2}$. 
    To examine this case, let $\cA_{\psi_1} = (Q_{1}, \pi_{1}, (\alpha_{P})_{P \subseteq \Pi_1}, F_{1}, F'_1)$ and let $\cA_{\psi_2} = (Q_{2}, \pi_{2}, (\beta_{P})_{P \subseteq \Pi_2}, F_{2}, F'_2)$.
    We construct the automaton $\cA_{\psi_1 \land \psi_2} = (Q, \pi, (\delta_{P})_{P \subseteq \Pi'}, F, F')$ over $\Pi' = \Pi \cup \Pi_{\cY_1 \cup \cY_2}$ by concatenating $\cA_{\psi_1}$ and $\cA_{\psi_2}$ as follows.
    The set of states is $Q = Q_{1} \times Q_{2}$.
    Let $P \subseteq \Pi'$, $P_1 = P \cap \Pi_{1}$ and $P_2 = P \cap \Pi_{2}$. 
    The initialization function is defined as follows: $\pi(P) = (q_{1}, q_{2})$ where $q_{1} = \pi_{1}(P_1)$ and $q_{2} = \pi_{2}(P_2)$.
    The transition function is defined as follows. Let $M = \{\{(q^{1}_{1}, q^{1}_{2}), \dots, (q^{m}_{1}, q^{m}_{2})\}\} \in \cM(Q)$. Now, $\delta_{P}(M) = (q_{1}', q_{2}')$, where $q_{1}' = \alpha_{P_1}(\{\{q^{1}_{1}, \dots, q^{m}_{1}\}\})$ and $q_{2}' = \beta_{P_2}(\{\{q^{1}_{2}, \dots, q^{m}_{2}\}\})$.
    The set of accepting states is $F = F_{1} \times F_{2}$, whereas the set of rejecting states is $F' = (F'_1 \times Q_2) \cup (Q_1 \times F'_2)$.

    Now, it is easy to show by induction on $m \in \N$ that for all rooted trees $(T, w)$ and for all interpretations $(T^*, w)$ of $(T, w)$ over $\cY_1 \cup \cY_2$, $\cA_{\psi_1 \land \psi_2}$ is in the state $(q_1, q_2)$ in round $m$ at $(T^*, w)$ iff $\cA_{\psi_2}$ is in the state $q_1$ and $\cA_{\psi_2}$ is in the state $q_2$ in round $m$ at $(T^*, w)$. Thus, for all rooted trees $(T,w)$ over $\Pi$ and all interpretations $(T^*, w)$ of $(T, w)$ over $\cY_1 \cup \cY_2$, $T^* \models \psi_1 \land \psi_2$ iff $T^* \models \psi_1 $ and $T^* \models \psi_2$ iff $(T^*, w)$ is fixed point accepted by $\cA_{\psi_1}$ and $(T^*, w)$ is fixed point accepted by $\cA_{\psi_2}$ (by the induction hypothesis) iff $(T^*,w)$ is fixed point accepted by $\cA_{\psi_1 \land \psi_2}$.
    Similarly, for all rooted trees $(T,w)$ and all interpretations $(T^*, w)$ of $(T, w)$ over $\cY_1 \cup \cY_2$, $T^* \not\models \psi_1 \land \psi_2$ iff $T^* \not\models \psi_1 $ or $T^* \not\models \psi_2$ iff $(T^*, w)$ is fixed point rejected by $\cA_{\psi_1}$ or $(T^*, w)$ is fixed point rejected by $\cA_{\psi_2}$ (by the induction hypothesis) iff $(T^*,w)$ is fixed point rejected by $\cA_{\psi_1 \land \psi_2}$.

    Clearly, by the construction, the automaton $\cA_{\psi_1 \land \psi_2}$ is forgetful and quasi-acyclic.
    Lastly, we note that the bound of $\cA_{\psi_1 \land \psi_2}$ is the maximum of the bounds of $\cA_{\psi_1}$ and $\cA_{\psi_2}$.

    \item Case $\varphi \colonequals \exists y \psi$. We first concatenate $\cA_{\psi}$ with the automaton from Lemma \ref{properness automaton} for the lone variable $y$ to obtain an automaton $\cA_{1}$, which behaves the same as $\cA_{\psi}$, but is run in models that interpret $y$. We then construct a deterministic automaton $\cA_{2}$ where we modify the initialization function to ``guess'' whether $y$ is true in each node; runs where $y$ is true in more than one node are simply ignored utilizing the previous concatenation with the automaton that checks that $y$ is interpreted ``properly'', i.e., as a singleton.

    More formally, we define the automaton $\cA_{1} = (Q_{1}, \pi_{1}, (\delta_{P})_{P \subseteq \Pi \cup \Pi_{\cY \cup \{y\}}}, F_{1}, F'_1)$ 
    as follows.
    First, let $\cA_{\psi} = (Q_{\psi}, \pi_{\psi}, (\alpha_{P})_{P \subseteq \Pi \cup \Pi_{\cY}}, F_{\psi}, F'_{\psi})$ be the automaton for $\psi$, and
    let $\cA_{\text{proper } y}  = (Q_{y}, \pi_{y}, (\beta_{P})_{P \subseteq \Pi \cup \Pi_{\{y\}}}, F_{y}, F'_y)$ be the automaton obtained by Lemma~\ref{properness automaton} for variable $y$, i.e, $\cA_{\text{proper } y}$ checks that $y$ is labeled correctly.
    We concatenate $\cA_{\psi}$ with $\cA_{\text{proper }y}$ as follows.
    \begin{itemize}
        \item The set of states is $Q_{1} \colonequals Q_{\psi} \times Q_{y}$.
        \item The initialization function $\pi_{1} \colon \Pi \cup \Pi_{\cY \cup \{y\}} \to Q_{1}$ is defined as follows. Let $P \subseteq \Pi \cup \Pi_{\cY \cup \{y\}}$. Now $\pi_{1}(P) = (\pi_{\psi}(P \cap (\Pi \cup \Pi_{\cY})), \pi_{y}(P \cap (\Pi \cup \Pi_{\{y\}})))$.
        \item Let $P \subseteq \Pi \cup \Pi_{\cY \cup \{y\}}$. The transition function $\delta_P \colon \cM(Q_{1}) \to Q_{1}$ is defined as follows. First, let $M = \{\{(q_{1}, q_{1}'), \dots, (q_{m}, q_{m}')\}\} \in \cM(Q_{1})$. Now $\delta_{P}(M) = (r, r')$ where $r = \alpha_{P \cap (\Pi \cup \Pi_{\cY})}(\{\{q_{1}, \dots, q_{m}\}\})$ and $r' = \beta_{P \cap (\Pi \cup \Pi_{\{y\}})}(\{\{q_{1}', \dots, q_{m}'\}\})$.
        \item The set of accepting states is $F_{1} = \{\, (q, q') \in Q_{1} \mid q \in F_{\psi}, q' \in F_{y} \,\}$ and the set of rejecting states is $F'_{1} = \{\, (q, q') \in Q_{1} \mid q \in F'_{\psi},\, q' \in F_y \,\}$. 
    \end{itemize}
    Now, (by similar arguments as for the conjunction) for all rooted trees $(T, w)$ and for all interpretations $(T^*,w)$ over $\Pi \cup \Pi_{\cY \cup \{y\}}$: clearly $\cA_1$ fixed point accepts $(T^*, w)$ if and only if $\cA_{\psi}$ fixed point accepts $(T^*, w)$ and $\cA_{\text{proper } y}$ fixed point accepts $(T^*, w)$, and $\cA_1$ fixed point rejects $(T^*, w)$ if and only if $\cA_{\psi}$ fixed point rejects $(T^*, w)$ and $\cA_{\text{proper } y}$ fixed point accepts $(T^*, w)$. 
    It is also easy to see that $\cA_1$ is forgetful and quasi-acyclic
    over interpretations over $\Pi \cup \Pi_{\cY \cup \{y\}}$ since $\cA_{\psi}$ and $\cA_{\text{proper }y}$ are forgetful and quasi-acyclic. 
    Clearly the bound of $\cA_{1}$ is the maximum of the bounds of $\cA_{\psi}$ and $\cA_{\text{proper } y}$.

    We then define a deterministic automaton $\cA_{2} = (Q_{2}, \pi_{2}, (\delta'_{P})_{P \subseteq \Pi \cup \Pi_{\cY \setminus \{y\}}}, F_{2}, F_{2}')$ that guesses whether a node is labeled $y$ in every round, and keeps track of all possible outcomes. It is defined as follows.
    \begin{itemize}
        \item The set of states is $Q_{2} = \cP(Q_{1})$.
        \item The initialization function is defined by $\pi_{2}(P) = \{\, \pi_{1}(P), \pi_{1}(P \cup \{y\}) \,\}$,
        \item The transition function is defined such that $\delta_{P}'(M)$ is the set of all $\delta_{P}(M')$ and $\delta_{P \cup \{y\}}(M')$, where $M'$ is obtained from $M$ by picking one state in $Q_{1}$ from each element of $M$.
        \item The set of accepting states is $F_{2} = \{\, Q \in Q_{2} \mid q \in Q \text{ for some } q \in F_{1} \,\}$.
        \item The set of rejecting states is $F_{2}' = \{\, Q \in Q_{2} \mid q \notin Q \text{ for all } q \in F_{1} \text{ and } q' \in Q \text{ for some } q' \in F_{1}' \,\}$.
    \end{itemize}

    Clearly by the construction $\cA_2$ is forgetful, since $\cA_1$ is also. 
    It is easy to see that $\cA_2$ is also quasi-acyclic as follows. Its states can be ordered by following the ordering of their components.
    In round $n$, the automaton $\cA_2$ has scanned the tree up to the depth $n$ from node $v$. Therefore, each ``guess'' for the variable $y$ at round $i$ can be also guessed at round $i+1$. 
    
    Next, we will show for all interpretations $(T^*, w)$ over $\Pi \cup \Pi_{\cY \setminus \{y\}}$ that 
    \begin{enumerate}
        \item $\cA_{1}$ fixed point accepts $(T^*[v/y], w)$ for some $v$ iff $\cA_{2}$ fixed point accepts $(T^*, w)$,
        \item $\cA_{1}$ fixed point rejects $(T^*[v/y], w)$ for all $v$ iff $\cA_{2}$ fixed point rejects $(T^*, w)$. 
    \end{enumerate}
    
    If $\cA_{1}$ fixed point accepts $(T^*[v/y], w)$ for some $v$, then clearly $\cA_{2}$ fixed point accepts $(T^*, w)$, since the accepting fixed point is always in the state of $w$ starting from some round. For converse, assume that $\cA_{2}$ fixed point accepts $(T^*, w)$. Let $n$ be the first round where $\cA_2$ reaches a fixed point accepting state. The correct interpretation for $y$ can be constructed as follows. We start from the root and pick an accepting state that appears in the set of states at round $n$. For nodes at level $1$, we pick a state for each node from round $n-1$ which gives the accepting state chosen for the root at round $n$. We continue this process up to the level $n$. Now, the correct labels can be chosen by first looking at the initial states of the nodes at level $n$. Then, the labels for the nodes at level $n-1$ can be chosen and so on. The labeling has to be correct since for the root we chose an accepting state. 

    If $\cA_{1}$ fixed point rejects $(T^*[v/y], w)$ for all $v$, then clearly $\cA_{2}$ fixed point rejects $(T^*, w)$, since the rejecting fixed point is always in the state of $w$ starting from some round and there is no round where an accepting state would appear. For the converse, assume that $\cA_{2}$ fixed point rejects $(T^*, w)$.
    Let $n$ be a round where $\cA_2$ is in a fixed point rejecting state. By using a similar arguments as in the case of accepting, we can construct an interpretation for $y$. We notice that every possible interpretation is invalid, since otherwise $\cA_2$ would fixed point accept. Thus, $\cA_1$ fixed point rejects for every possible interpretation.
    \item Case $\varphi \colonequals \exists Y \psi$. The same idea as with the first-order variables except that we do not have to concatenate the automaton with $\cA_{\text{proper } y}$.
\end{enumerate}
\end{proof}

We now create a non-deterministic automaton without rejecting states that computes $\cA_{\varphi}$ with an ordinary accepting condition.

\begin{lemma}\label{lemma: non-deterministic MSO automaton}
    Let $\varphi$ be an $\MSO$-formula definable by an $\VGML$-formula. We can construct a non-deterministic carefree bounded $\FCMPA$ $\cA$ such that for every rooted tree $(T, w)$: $T \models \varphi(w)$ if and only if there is a $k \in \N$ such that $\cA$ $k$-omnipresently accepts $(T, w)$.
\end{lemma}
\begin{proof}
    Let $\cA_{\varphi} = (Q, \pi, (\delta_{P})_{P \subseteq \Pi}, F, F')$ be the deterministic carefree bounded $\FCMPA^{r}$ constructed in Theorem \ref{theorem: Buchi distributed MSO automaton}. Let $q_{(T, w)}$ denote the fixed-point state of $\cA_{\varphi}$ at the root $w$ in the tree $T$. Let $Q_{P} = \{\, q_{(T, w)} \in Q \mid V(w) = P \,\}$, i.e., $Q_{P}$ is the set of all fixed-point states that $\cA_{\varphi}$ can visit in a node labeled with the proposition symbols in $P$. We construct the non-deterministic automaton $\cA = (Q, \pi', (\delta_{P})_{P \subseteq \Pi}, F)$ where $\pi'(P) = Q_{P}$ for all $P \subseteq \Pi$. In other words, we drop the rejecting states and initialize each node with its fixed point state as the root of some tree. It is easy to see that $\cA$ is carefree and bounded, as the transition function is unchanged from $\cA_{\varphi}$. Now we are ready to prove the equivalence.

    ``$\Rightarrow$'' 
    Assume $T \models \varphi(w)$. Let $\bigvee_{i \in \N} \psi_{i}$ be the $\VGML$-formula equivalent to $\varphi$. 
    Therefore, by Proposition \ref{proposition: new prefix tree and extensions 2} for some $k \in \N$, $(T, w)$ is $k$-extendable w.r.t. $\varphi$. Let $T'$ be a $k$-extension of $T$.
    Thus, $\cA_{\varphi}$ fixed point accepts $(T', w)$ by Theorem \ref{theorem: Buchi distributed MSO automaton}. 
    For every node $v$ on level $k$ in $T$, it is easy to see that its fixed point in $T'$ is the same as its fixed point in $T'_{v}$ (the subtree of $T'$ rooted at $v$). Now if each node $v$ on level $k$ of $T$ chooses its fixed point in $T'_{v}$ in the initial round, it is clear that $w$ visits an accepting state in round $k$ in $T$. Since this holds for every $k$-extension of $T$, and since each state in $Q_{P}$ is the fixed point in some possible extension, it is clear that every run of $\cA$ accepts $w$ in $T$ in round $k$.

    ``$\Leftarrow$'' 
    Assume that $k \in \N$ and every run of $\cA$ in $T$ accepts $w$ in round $k$. Consider a run of $\cA$, where for every node $v$ in $T$ on level $k$, we initialize $v$ by choosing the state $q_{v} \in Q_{V(v)}$ which is the fixed point of $\cA_{\varphi}$ at $v$ in $T_{v}$ (the subtree of $T$ rooted at $v$), which is the same as the fixed point of $\cA_{\varphi}$ at $v$ in $T$. We know that this run in $T$ accepts in round $k$, where the state of $w$ in round $k$ is clearly the fixed point of $\cA_{\varphi}$ at $w$ in $T$. Since this state is an accepting state, we know that $\cA_{\varphi}$ fixed point accepts $w$, which by 
    Theorem \ref{theorem: Buchi distributed MSO automaton} means that $T \models \varphi(w)$.
\end{proof}

By Lemma \ref{lemma: non-deterministic MSO automaton} we can conclude Lemma \ref{MSO-definable properties} (which we recall first).

\textbf{Lemma \ref{MSO-definable properties}.}
\emph{Given an $\MSO$-formula $\varphi(x)$ definable by an $\VGML$-formula, we can construct a bounded carefree $\FCMPA$ $\cA$ such that for all rooted trees $(T, w)$ we have
    \[
    \text{$(T, w)$ is accepted by $\cA$} \iff T \models \varphi(w).
    \]
}

\begin{proof}
    Let $\cA'$ be the non-deterministic carefree bounded $\FCMPA$ constructed in Lemma \ref{lemma: non-deterministic MSO automaton}. By Lemma \ref{lemma: non-deterministic to deterministic}, we can construct a deterministic carefree bounded $\FCMPA$ $\cA$ that accepts a rooted tree $(T, w)$ if and only if $\cA'$ $k$-omnipresently accepts $(T, w)$ for some $k \in \N$. Thus by Lemma \ref{lemma: non-deterministic MSO automaton}, $\cA$ accepts a rooted tree $(T, w)$ if and only if $T \models \varphi(w)$.
\end{proof}

By Lemma \ref{MSO-definable properties} it is straightforward to conclude the main theorem.

\textbf{Theorem 4.1 in \cite{ahvonen_neurips}.}
\emph{Let $P$ be a property definable in $\MSO$. Then $P$ is expressible as a $\GNN[\R]$ if and only if it is definable in $\GMSC$.}
\begin{proof}
Let $P$ be an $\MSO$-definable (node) property over $\Pi$.

First, assume that a $\GNN[\R]$ $\cG$ defines $P$. Then by Theorem 3.4 in \cite{ahvonen_neurips} and Lemma \ref{MSO-definable properties}, there is a bounded $\FCMPA$ that defines $P$. Thus by Proposition 3.1 in \cite{ahvonen_neurips}, there is an equivalent $\GMSC$-program that defines $P$.

For the converse, assume that a $\GMSC$-program $\Lambda$ defines $P$. Now, by Proposition 2.4 in \cite{ahvonen_neurips}, there exists an equivalent $\VGML$-formula $\psi$ that defines $P$ and thus by Theorem 3.4 in \cite{ahvonen_neurips}, there is an equivalent $\GNN[\R]$ that defines $P$.  
\end{proof}

We conclude with a brief discussion about the framework where the fixed point accepting condition is used by $\CMPA$s and $\GNN$s. In this framework, it is clear that $\GNNF$s have the same expressive power as bounded $\FCMPA$s and that $\GNN[\R]$s have the same expressive power as $\CMPA$s; this follows from the translations in \cite{ahvonen_neurips}. We then obtain the following result as a corollary of Theorem \ref{theorem: Buchi distributed MSO automaton}.
\begin{theorem}
    In restriction to $\MSO$, each $\GNN[\R]$ with the fixed point accepting condition translates to an equivalent carefree $\GNNF$ with the fixed point accepting condition.
\end{theorem}
The above theorem also holds when the $\GNN[\R]$s have any other accepting condition.
The other direction, that each $\GNNF$ with the fixed point accepting condition translates to an equivalent $\GNN[\R]$ with the fixed point accepting condition is trivial and also applies in restriction to MSO. 
This direction also holds when both $\GNNF$s and $\GNN[\R]$s have other acceptance conditions, as long as they both have the same accepting condition.
Both directions hold even when we restrict to the cases where the $\GNN$s involved are quasi-acyclic, forgetful, or both.
On another note, it seems clear that counting bisimulation-invariant MSO has the same expressive power as the graded modal $\mu$-calculus over finite models, though this result has not been formally stated in the literature. If this is indeed the case, then we obtain the following result by the fact that the graded modal $\mu$-calculus is a fragment of MSO:
for each formula of the graded modal $\mu$-calculus, we can construct an equivalent $\GNNF$ and $\GNN[\R]$ with the fixed point accepting condition.
Conversely, since $\GNN$s are counting bisimulation-invariant, it would follow that in restriction to MSO, each $\GNNF$ and $\GNN[\R]$ with the fixed point accepting condition translates to an equivalent formula of the graded modal $\mu$-calculus.

\section{Conclusion}

We have given an alternative proof for an existing result that recurrent graph neural networks working with real numbers have the same expressive power relative to MSO as the graded modal substitution calculus. The proof involves distributed message passing automata rather than parity tree automata as used in the existing proof. We have also discussed some variants of accepting conditions.

\bibliography{references}

\begin{thebibliography}{1}

\bibitem{ahvonen2025classdistributedautomatacontains}
Veeti Ahvonen, Damian Heiman, and Antti Kuusisto.
\newblock A class of distributed automata that contains the modal mu-fragment, 2025.
\newblock URL: \url{https://arxiv.org/abs/2505.07816v1}, \href {https://arxiv.org/abs/2505.07816v1} {\path{arXiv:2505.07816v1}}.

\bibitem{ahvonen_neurips}
Veeti Ahvonen, Damian Heiman, Antti Kuusisto, and Carsten Lutz.
\newblock Logical characterizations of recurrent graph neural networks with reals and floats.
\newblock In Amir Globersons, Lester Mackey, Danielle Belgrave, Angela Fan, Ulrich Paquet, Jakub~M. Tomczak, and Cheng Zhang, editors, {\em Advances in Neural Information Processing Systems 38: Annual Conference on Neural Information Processing Systems 2024, NeurIPS 2024, Vancouver, BC, Canada, December 10 - 15, 2024}, 2024.

\bibitem{DBLP:conf/cade/KupfermanSV02}
Orna Kupferman, Ulrike Sattler, and Moshe~Y. Vardi.
\newblock The complexity of the graded {\(\mathrm{\mu}\)}-calculus.
\newblock In Andrei Voronkov, editor, {\em Automated Deduction - CADE-18, 18th International Conference on Automated Deduction, Copenhagen, Denmark, July 27-30, 2002, Proceedings}, volume 2392 of {\em Lecture Notes in Computer Science}, pages 423--437. Springer, 2002.
\newblock \href {https://doi.org/10.1007/3-540-45620-1\_34} {\path{doi:10.1007/3-540-45620-1\_34}}.

\bibitem{Kuusisto13}
Antti Kuusisto.
\newblock {Modal Logic and Distributed Message Passing Automata}.
\newblock In {\em Computer Science Logic 2013 (CSL 2013)}, volume~23 of {\em Leibniz International Proceedings in Informatics (LIPIcs)}, pages 452--468, 2013.

\end{thebibliography}

\end{document}